\documentclass[submission,copyright,creativecommons]{eptcs}
\providecommand{\event}{AFL 2023} 

\usepackage{iftex}
\usepackage[utf8]{inputenc}
\usepackage[english]{babel}
\usepackage{amsthm}
\usepackage{thm-restate, etoolbox}
\usepackage{amsmath,amsfonts,amssymb,dsfont,centernot}
\usepackage{stmaryrd,mathtools}
\usepackage{xcolor}
\usepackage{tikz}
\usepackage{forest,graphicx}
\usetikzlibrary{arrows,positioning,decorations.pathreplacing}

\allowdisplaybreaks

\useshorthands{"}
\addto\extrasenglish{\languageshorthands{ngerman}}
\DeclareMathOperator{\wt}{wt}
\DeclareMathOperator{\pos}{pos}

\DeclareMathOperator{\he}{ht}
\DeclareMathOperator{\supp}{supp}

\DeclareMathOperator{\var}{var}
\DeclareMathOperator{\rk}{rk}

\providecommand*{\nat}[0]{\ensuremath{\mathbb N}}
\providecommand*{\seq}[3]{\ensuremath{#1_{#2}, \dotsc, #1_{#3}}}
\providecommand*{\abs}[1]{\ensuremath{\lvert #1 \rvert}}

\newtheorem{theorem}{Theorem}
\newtheorem{definition}[theorem]{Definition}
\newtheorem{lemma}[theorem]{Lemma}
\newtheorem{proposition}[theorem]{Proposition}
\newtheorem{example}[theorem]{Example}
\newtheorem{corollary}[theorem]{Corollary}
\newtheorem{remark}[theorem]{Remark}

\newcommand{\Tsigma}{T_{\Sigma}}
\newcommand{\B}{\mathbb{B}}
\newcommand{\SR}{\mathbb{S}}

\newcommand{\cA}{\mathcal{A}}
\newcommand{\cB}{\mathcal{B}}

\newcommand{\N}{\mathbb{N}}
\newcommand{\Z}{\mathbb{Z}}

\providecommand{\sem}[1]{\ensuremath{\llbracket #1 \rrbracket}}

\ifpdf
  \usepackage{underscore}         
  \usepackage[T1]{fontenc}        
\else
  \usepackage{breakurl}           
\fi

\title{Solving the Weighted HOM-Problem\\With the Help of Unambiguity}
\author{Andreea-Teodora Nász
\institute{Leipzig University\\ Leipzig, Germany}
\institute{Faculty of Mathematics and Computer Science
\\
PO box 100 920, 04009 Leipzig, Germany}
\email{nasz@informatik.uni-leipzig.de} }
\def\titlerunning{Unambiguity in the Weighted HOM-Problem}
\def\authorrunning{Andreea-Teodora Nász}
\begin{document}
\maketitle

\begin{abstract}
The \emph{HOM-problem}, which asks whether the image of a regular tree language under a tree homomorphism is again regular, is known to be decidable by [Godoy, Gim\'enez, Ramos, 
\`Alvarez: The HOM problem is decidable. STOC (2010)]. 
Research on the weighted version of this problem, however, is still in its infancy since it requires customized investigations.
In this paper we address the weighted HOM-problem and strive to keep the underlying semiring as general as possible. In return, we restrict the input:
We require the tree homomorphism~$h$ to be~\emph{tetris-free}, a condition weaker than injectivity, 
and for the given weighted tree automaton, we propose an ambiguity notion with respect to~$h$. These assumptions suffice to ensure decidability of the thus restricted HOM-problem for all zero-sum free semirings by allowing us to reduce it to the (decidable) unweighted case.
\end{abstract}

\section{Introduction}

Over the past decades, various extensions to the well-known model of finite-state automata have 
been proposed. These acceptors were taken to the next level when their qualitative evaluation was 
generalized to a quantitative one, which led to the concept of~\emph{weighted 
automata}~\cite{schutzenberger1961}. Such devices assign a weight to each input word, 
thus computing so-called~\emph{formal power series}. Weighted automata are 
commonly used to model numerical factors related to the input, such as costs, 
probabilities and consumption of resources or time, and enjoy consistent attention from the research community focused on automata 
theory~\cite{droste2021weighted,salomaa2012automata}. 
The favored algebraic structure for performing weight calculations are 
semirings~\cite{gol99,hebisch1998semirings}, as they are quite general while still being computationally efficient due to their 
distributivity.

Another dimension for generalizing finite-state automata lifts their input to more 
complex data structures such as infinite words~\cite{altinfinitewords,infinitewords}, 
trees~\cite{tataok}, graphs~\cite{graphautomata} and 
pictures~\cite{giammarresi1992recognizable,rosenfeld2014picture}. Particularly,
\emph{finite-state tree automata} were introduced independently 
in~\cite{doner1970tree,thatcher1965generalized,thatcher1968generalized}. The 
so-called~\emph{regular tree languages} they recognize have been studied 
extensively~\cite{tataok}, 
and find applications in a variety of areas like natural language processing~\cite{jurmar08}, 
picture generation~\cite{drewes2006grammatical} and compiler construction~\cite{wilseihac13}.
In many cases, applications require both types of generalizations, and so several 
models of~\emph{weighted tree automata} (WTA) and the~\emph{regular tree series} they 
recognize continue to be studied~\cite{fulvog09}. 

The price to pay for the simplicity of tree automata lies in their significant limitations.
For instance, they cannot ensure that certain subtrees of input trees are equal~\cite{gecste15}, much 
like the classical (string) automata cannot ensure that the number of~$a$'s and~$b$'s in a word 
is equal. This defect was tackled with extensions proposed in~\cite{rateg1981} 
and~\cite{hom2012exp,godoy2013hom,godoy2010hom} where~\emph{tree automata with constraints} 
can explicitly 
require or forbid certain subtrees to be equal. Such devices have played a crucial part in deciding 
the~\emph{HOM-problem}: This long-standing open question~\cite{tataok}
asks, given a regular tree language and a tree homomorphism, whether the image is again regular. 
A tree homomorphism performs a transformation on trees and can duplicate subtrees, therefore the 
trees in the homomorphic image might have certain identical subtrees, which calls for 
the constraints mentioned above. In~\cite{godoy2013hom}, 
the authors first represent this homomorphic image of a regular tree language by a tree automaton 
with explicit constraints, and then decide algorithmically if the language it recognizes is regular 
despite the constraints it imposes.

The nature of the \emph{weighted} HOM-problem, where a regular tree series and a tree homomorphism 
are given as input, requires an individual investigation for different semirings. Recently, the approach 
from~\cite{godoy2013hom} was adjusted to the special case of nonnegative 
integers~\cite{integer-hom}, but so far, the question remains open for other semirings. In this 
paper, we reverse the strategy and impose conditions on the input in order to decide the thus 
restricted HOM-problem for a larger class of semirings. More precisely, we require our 
protagonist -- the weighted tree automaton with constraints -- to be unambiguous, and reduce the 
question of its regularity to the unweighted case from~\cite{godoy2013hom} for any zero-sum free 
(commutative) semiring. Afterwards, we phrase a condition on the input of the HOM-problem which 
ensures that our reduction is applicable.

This article consists of five sections including its introduction. Our main contributions can be 
summarized as follows:
\begin{itemize}
\item In Section~\ref{sec:prelim} we establish notations and recall the main objects that will play a 
role throughout the paper, primarily the~\emph{weighted tree automata with hom-constraints (WTAh)} 
which are used to represent homomorphic images of regular tree series. 
\item In Section~\ref{sec:reg decid} we prove that regularity is decidable for the unambiguous 
devices of this type over zero-sum free semirings. We achieve this by reducing the 
question to the unweighted case where regularity is known to be decidable~\cite{godoy2013hom}.
\item In Section~\ref{sec:sufficient} we integrate this decidability result into the HOM-problem. To 
this end, we phrase a condition on the input of the HOM-problem which guarantees that the WTAh
constructed for this instance is unambiguous. Thus, the HOM-problem with input restricted 
accordingly is decidable for any zero-sum free semiring.
\item Finally, in Section~\ref{sec:conclusion} we briefly summarize our results and 
discuss further research that will extend the present work.
\end{itemize}

\section{Preliminaries and Technical Background}\label{sec:prelim}
We begin as usual with the necessary background for this paper.
\subsection*{General Notation}
We denote the set~$\{0,1,2,\ldots\}$ of nonnegative integers by $\N$, and we
let $[k]=\{1,\ldots,k\}$ for every~$k\in\N$. Let~$A$ and $B$ be sets. 
We write~$\vert A\vert$ for the cardinality of $A$, and~$A^*$ for the set of 
finite strings over~$A$. The empty string is~$\varepsilon$ and the length of 
a string~$w$ is~$\vert w\vert$. For a mapping~$f\colon A\to B$ and
$S\subseteq B$ we denote the inverse image of~$S$ under~$f$ by~$f^{-1}
(S)$, and we write~$f^{-1}(b)$ instead of~$f^{-1}(\{b\})$ for every~$b\in B$. 

\subsection*{Trees}
A~\emph{ranked alphabet} is a pair~$(\Sigma,\rk)$ that consists of a finite set~$\Sigma$
and a rank mapping~$\rk\colon\Sigma\to\N$. For every~$k\geq 0$, we define~$\Sigma_k=
\rk^{-1}(k)$, and we sometimes write~$\sigma^{(k)}$ to indicate  that~$\sigma\in\Sigma_k$.
We often abbreviate~$(\Sigma,\rk)$ by~$\Sigma$ leaving~$\rk$ implicit. 
Let~$Z$ be a set disjoint with~$\Sigma$. The set of~\emph{$\Sigma$-trees over~$Z$},
denoted~$T_\Sigma(Z)$, is the smallest set~$T$ such that~(i) $\Sigma_0\cup Z\subseteq T$
and~(ii) $\sigma(\seq t1k)\in T$ for every~$k\in\N$, $\sigma\in\Sigma_k$ and~$\seq t1k
\in T$. We abbreviate~$T_\Sigma(\emptyset)$ simply to~$T_\Sigma$, and call any
subset~$L\subseteq T_\Sigma$ a \emph{tree language}.
Consider~$t\in T_\Sigma(Z)$. The set~$\pos(t)
\subseteq\N^*$ of~\emph{positions of~$t$} is defined inductively by~$\pos(t)=\varepsilon$ 
for every~$t\in \Sigma_0\cup Z$, and by $$\pos\big(\sigma(\seq t1k)\big)=\{\varepsilon\}
\cup\bigcup_{i\in [k]} \{ip\mid p\in\pos(t_i)\}$$ for all~$k\in\N$, $\sigma\in\Sigma_k$ 
and~$\seq t1k\in T_\Sigma(Z)$. The set of positions of~$t$ inherits the lexicographic 
order~$\leq_{\text{lex}}$ from~$\N^*$.
The~\emph{size~$\vert t\vert$} and \emph{height~$\he(t)$ 
of~$t$} are defined as \[\vert t\vert \;=\;\vert \pos(t)\vert \quad\text{ and }\quad 
\he(t)=\max_{p\in\pos(t)} \vert p\vert\;.\] 
For~$p\in\pos(t)$, the~\emph{label~$t(p)$ of~$t$ at~$p$}, the~\emph{subtree~$t|_p$ 
of~$t$ at~$p$} and the~\emph{substitution~$t[t']_p$ of~$t'$ into~$t$ 
at~$p$} are defined
\begin{itemize}
\item for~$t\in\Sigma_0\cup Z$ by $t(\varepsilon)=t|_\varepsilon=t$ and~$t[t']_\varepsilon=t'$, and
\item for $t=\sigma(\seq t1k)$ by $t(\varepsilon)=\sigma$, $\;t(ip')=
t_i(p')$, $\;t|_\varepsilon=t$, $\;t|_{ip'}=t_i|_{p'}$, $\;t[t']_\varepsilon=t'$, and $$t[t']_{ip'}
=\sigma(\seq t1{i-1},t_i[t']_{p'},\seq t{i+1}k)$$ for all~$k\in\N$, $\;\sigma\in\Sigma_k$, $\;\seq 
t1k\in T_\Sigma(Z)$, $\;i\in[k]$ and~$p'\in\pos(t_i)$.
\end{itemize}
For every subset~$S\subseteq\Sigma
\cup Z$, we let~$\pos_S(t)=\{p\in\pos(t)\mid t(p)\in S\}$ and we abbreviate~$\pos_{\{s\}}
(t)$ by~$\pos_s(t)$ for every~$s\in\Sigma\cup Z$.
Let~$X=\{x_1,x_2,\ldots\}$ be a fixed, countable set of formal variables. For $k\in\N$
we denote by~$X_k$ the subset $ \{\seq x1k\}$. For any~$t\in T_\Sigma(X)$ we let $$\var(t)=
\{x\in X\mid \pos_x(t)\neq\emptyset\}\;.$$ Finally, for~$t\in T_\Sigma(Z)$, a subset~$V\subseteq 
Z$ and a mapping~$\theta\colon V\to T_\Sigma(Z)$, we define the \emph{substitution~$t\theta$ 
applied to~$t$} by~$v\theta=\theta(v)$ for~$v\in V$, $\;z\theta =z$ 
for~$z\in Z\setminus V$, and $$\sigma(\seq t1k)\theta=\sigma\big(t_1\theta,\ldots,t_k
\theta\big)$$ for all~$k\in\N$, $\sigma\in\Sigma_k$ and~$\seq t1k\in T_\Sigma(Z)$. 
If~$V=\{\seq v1n\}$, we write the substitution~$\theta$ explicitly as~$[v_1 \leftarrow 
\theta(v_1),\ldots,v_n\leftarrow\theta(v_n)]$, and abbreviate it further to~$[\theta(x_1),
\ldots,\theta(x_n)]$ if~$V=X_n$.

\subsection*{Semirings and Tree Series}
A \emph{(commutative) semiring}~\cite{gol99,hebwei98} is a tuple~$(\SR, \mathord+, 
\,\cdot\:, 0, 1)$ such that~$(\SR, \mathord+, 0)$~and $(\SR,\, \cdot\:, 1)$ are
commutative monoids, $\cdot\:$~distributes over~$+$, and~$0 \cdot s = 0$ for all~$s\in 
\SR$.  Examples include 
\begin{itemize}
\item the Boolean semiring~$\B = \bigl(\{ 0,1 \}, \mathord\vee, 
\mathord\wedge,0, 1 \bigr)$,
\item the semiring~$\N = \bigl(\N , \mathord+,\mathord\cdot\:, 
0, 1\bigr)$,
\item the semiring~$\Z = \bigl(\Z , \mathord+,\mathord\cdot\:,
0, 1\bigr)$,
\item the tropical semiring $\mathbb T =\bigl(\N \cup\{\infty\}, \mathord{\min}, 
\mathord+, \infty, 0 \bigr)$, and
\item the arctic semiring $\mathbb A = \bigl(\N \cup
\{-\infty\}, \mathord{\max}, \mathord+, -\infty, 0 \bigr)$.
\end{itemize} 
When there is no risk of confusion, we refer to a semiring~$(\SR, 
\mathord+, \,\mathord\cdot\:, 0, 1)$ simply by its 
carrier set~$\SR$. We call~$\SR$~\emph{zero-sum free} if~$a+b=0$
implies~$a=b=0$ for all~$a,b\in\SR$.  All semirings listed above except for~$\Z$ are zero-sum free. 
Let~$\Sigma$ be a ranked alphabet and~$Z$ a set. Any mapping~$\varphi \colon T_\Sigma 
(Z)\to \SR$ is called a~\emph{tree series} or \emph{weighted tree language} over~$\SR$, and 
its~\emph{support} is the set~$\supp(\varphi) = \{t \in T_\Sigma(Z) \mid \varphi(t) \neq 0\}$. 

\subsection*{Tree Homomorphisms}
Given ranked alphabets $\Sigma$~and~$\Delta$, let~$h' \colon
\Sigma \to T_\Delta(X)$ be a mapping that satisfies $h'(\sigma) \in
T_\Delta(X_k)$ for all~$k \in \N$ and~$\sigma \in \Sigma_k$.  We
extend~$h'$ to~$h \colon T_\Sigma \to T_\Delta$ by $h(\alpha) =
h'(\alpha) \in T_\Delta(X_0) = T_\Delta$ for all~$\alpha \in \Sigma_0$
and $$h(\sigma(\seq s1k)) = h'(\sigma)[x_1 \gets h(s_1), \dotsc, x_k
\gets h(s_k)]$$ for all~$k \in \N$, $\sigma \in \Sigma_k$, and~$\seq
s1k \in T_\Sigma$.  The mapping~$h$ is called the \emph{tree
  homomorphism induced by~$h'$}, and we identify~$h'$ and its induced 
tree homomorphism~$h$.  We call~$h$ 
\begin{itemize}
\item \emph{nonerasing} if~$h(\sigma) \notin X$ for all~$\sigma\in\Sigma$, 
\item \emph{nondeleting} if~$\sigma \in \Sigma_k$ implies~$\var(h'(\sigma)) = X_k$ for all~$k\in\N$, 
and
\item \emph{input-finitary} if the preimage~$h^{-1}(t)$ is finite 
for every~$t \in T_\Delta$.
\end{itemize} 
If a tree homomorphism~$h \colon T_\Sigma \to T_\Delta$ is nonerasing and 
nondeleting, then for every~$s\in h^{-1}(t)$, it is~$\abs s \leq \abs t$. In particular, 
$h$~is then~input finitary. 

Consider a tree series~$A \colon \Tsigma \to \SR$. Its~\emph{homomorphic image 
under~$h$} is the tree series~$h_A \colon T_\Delta \to \SR$ defined for every~$t \in
T_\Delta$ by $$h_A(t) = \sum_{s \in h^{-1}(t)} A(s)\;.$$ This definition relies on the tree
homomorphism to be input-finitary, otherwise the above sum is 
not finite, so the value~$h_A(t)$ might not be well-defined. For this reason, we 
will only consider nondeleting and nonerasing tree homomorphisms.

\subsection*{Weighted Tree Automata with Constraints}
Recently it was shown~\cite{dlt22,WTAc-journal} that such homomorphic images of regular tree 
languages can be represented efficiently using~\emph{weighted tree automata with 
hom-constraints} 
(WTAh). These devices were first introduced for the unweighted case in~\cite{godoy2013hom} and defined for zero-sum free commutative semirings in~\cite{dlt22}.%
\begin{definition}[cf.~\protect{\cite[Definition~1]{WTAc-journal}}]
Let~$\SR$ be a commutative semiring. A \emph{weighted tree automaton 
over~$\SR$ with hom-constraints (WTAh)} is a tuple~$\cA=\big(Q,\Sigma,F,R,\wt\big)$ such 
that~$Q$ is a finite set of states, $\Sigma$~is a ranked alphabet, $F\subseteq Q$ is the set 
of final states, $R$ is a finite set of rules of the form~$(\ell,q,E)$ such that $\ell\in 
T_\Sigma(Q)\setminus Q$, $q\in Q$ and~$E$ is an equivalence relation on~$\pos_{Q}(\ell)$, 
and~$\wt\colon R\to \SR$ assigns a weight to each rule.
\end{definition}

Rules of a WTAh are typically depicted as~$r=\ell \stackrel{E}\longrightarrow_{\wt(r)} q$. The 
components of such a rule are the~\emph{left-hand side}~$\ell$, the~\emph{target 
state}~$q$, the set~$E$ of~\emph{hom-constraints} and the~\emph{weight}~$\wt(r)$.
A hom-constraint~$(p,p')\in E$ is listed as~``$\:p=p'\;$'', and if~$p$ and~$p'$ are distinct, 
then~$p,\,p'$ are called~\emph{constrained positions}. The equivalence class of~$p$ in~$E$ is 
denoted~$[p]_{\equiv_E}$. We typically omit the trivial constraints~$(p,p)\in E$.

\begin{example}\label{ex:first ex}
Let~$\Sigma$ be the ranked alphabet~$\{a^{(0)},g^{(1)},k^{(2)}\}$. Consider the WTAh~$\cA=
\big(Q,\Sigma,F,R,\wt\big)$ over~$\Z$ with~$Q=\{q,q_f\}$,  $F=\{q_f\}$ and the set of rules and 
weights
\begin{align*}
R\;=\; \big\{\; a\to_1 q\,,\quad g(q)\to_2 q\,,\quad k\big(q,g(q)\big)
\stackrel{1=21}\longrightarrow_1 q_f \;\}\,.
\end{align*}
The only constrained positions are~$1$ and~$21$ in the rule with left-hand 
side~$k\big(q,g(q)\big)$.
\end{example}

The WTAh is a~\emph{weighted tree grammar} (WTG) if~$E$ is the identity relation for every rule~$\ell
\stackrel{E}\longrightarrow q$, and a WTA in the classical sense~\cite{tataok} if additionally 
$\pos_\Sigma(\ell)=\{\varepsilon\}$. WTG and WTA are equally expressive, as WTG can be translated 
straightforwardly into WTA using additional states. 

In this work, we are particularly interested in a specific subclass of WTAh, namely 
the~\emph{eq-restricted} WTAh~\cite{WTAc-journal}. In such a device, there is a 
designated~\emph{sink-state} whose sole purpose is to neutrally process copies of identical subtrees. 
More precisely, whenever different subtrees are mutually constrained, there is one leading copy among 
them that can be processed with arbitrary states and weights, while every other copy is handled 
exclusively by the weight-neutral sink-state.

\begin{definition}
A WTAh~$\big(Q,\Sigma,F,R,\wt\big)$ is~\emph{eq-restricted} if it has a~\emph{sink 
state}~$\bot\in Q\setminus F$ such that
\begin{itemize}
\item for all~$\sigma\in\Sigma$, the rule~$\sigma(\bot,\ldots,\bot)\to_1 \bot$ belongs
to~$R$, and no other rule targets~$\bot$, and
\item for every rule~$\ell\stackrel{E}
\longrightarrow q$ with~$q\neq\bot$, the following conditions hold: \\
Let~$\pos_{Q}(\ell)=\{\seq p1n\}$ and~$q_i=\ell(p_i)$ for~$i\in[n]$.
\begin{enumerate}
\item For each~$i\in[n]$, there exists~$q'\in Q\setminus\{\bot\}$ 
with~$\{q_j\mid p_j\in [p_i]_{\equiv_E}\}\setminus \{\bot\}=\{q'\}$.
\item There exists exactly one~$ p_j\in [p_i]_{\equiv_E}$ such that~$q_j=q'$.
\end{enumerate}
\end{itemize}
\end{definition} 
In other words, among each $E$-equivalence class of positions of a left-hand 
side~$\ell$, there is only one occurrence of a state different from~$\bot$, every other 
related position is labelled by~$\bot$. Moreover, $\bot$ processes every possible tree with 
weight~$1$. Whenever we consider an eq-restricted WTAh, we denote its state set 
by~$Q\dot{\cup}\{\bot\}$ instead of~$Q\ni \bot$ to point out the sink-state. 

\begin{example}\label{ex:second ex}
Recall the WTAc~$\cA$ from Example~\ref{ex:first ex}. It is not eq-restricted since the constrained 
positions~$1$ and~$21$ are both labeled by the same state, which is not a sink state. Instead, let 
us add a non-final state~$\bot$ to~$Q$, replace the rule~$k\big(q,g(q)\big)\stackrel{1=21}
\longrightarrow_1 q_f$ with~$k\big(q,g(\bot)\big)\stackrel{1=21}\longrightarrow_1 q_f$ and add 
the required rules targeting~$\bot$ to obtain an eq-restricted WTAh~$\cA'$. More precisely, we 
have the eq-restricted WTAh~$\cA'=\big(\{q,q_f,\bot\},\Delta,\{q_f\},R',\wt'\big)$ with the set of 
rules and weights
\begin{align*}
R'= \;&\big\{\, a\to_1 q,\quad \;g(q)\to_2 q,\;\;\,\quad k\big(q,g(\bot)\big) \stackrel{1=21}
\longrightarrow_1 q_f\, \big\} \\ 
\cup \;&\{\,a\to_1\bot, \quad g(\bot)\to_1\bot ,
\quad k(\bot,\bot)\to_1\bot\,\}\; .
\end{align*} 
\end{example}

Next, let us recall the semantics of WTAh from~\protect{\cite[Definitions~2~and~3]{WTAc-journal}}.
\begin{definition}
Let~$\cA=\big(Q,\Sigma,F,R,\wt\big)$ be a WTAh. A~\emph{run of~$\cA$} is 
a tree over the ranked alphabet~$\Sigma\cup R$ where the rank of a rule is~$\rk(\ell\stackrel{E}
\longrightarrow q)=\vert \pos_{Q}(\ell) \vert$, and it is defined inductively. Consider~$\seq t1n\in 
T_\Sigma$, $\seq q1n\in Q$ and suppose that~$\varrho_i$ is a run of~$\cA$ for~$t_i$ 
to~$q_i$ with weight~$\wt(\varrho_i)=a_i$ for each~$i\in[n]$. 
Assume that there is a rule of the form~$\ell\stackrel{E}\longrightarrow_a q$ in~$R$ such 
that~$\ell=\sigma(\seq {\ell}1m)$, $\pos_{Q}(\ell)=\{\seq p1n\}$ with~$\ell(p_i)=q_i$ and that for 
all~$p_i=p_j\in E$, it is~$t_i=t_j$. Then the following is a run of~$\cA$ for the tree 
$t=\ell[t_1]_{p_1}\dotsm[t_n]_{p_n}$ to~$q$: 
$$\varrho=\big(\ell\stackrel{E}\longrightarrow_a q\big)
(\seq {\ell}1m)[\varrho_1]_{p_1}\dotsm[\varrho_n]_{p_n}\;.$$
Its~\emph{weight~$\wt(\varrho)$} is computed as~$a\cdot 
\prod_{i\in [n]} a_i$. If~$\wt(\varrho)\neq 0$, then~$\varrho$ is~\emph{valid}, and if in addition,
$q\in F$ for its \emph{target state}~$q$, then $\varrho$ is~\emph{accepting}. 
We call~$\cA$~\emph{unambiguous} if for every~$t\in 
T_\Sigma$ there is at most one accepting run. The value~$\wt^q(t)$ is the sum of all 
weights~$\wt(\varrho)$ of runs of~$\cA$ for~$t$ to~$q$. Finally, the
tree series recognized by~$\cA$ is defined simply by
 $$\sem{\cA}\colon \;T_\Sigma\to\SR,\quad t\mapsto\sum_{q\in F} \wt^q(t)\; .$$
\end{definition}

Since the weights of rules are 
multiplied, we can assume wlog.\ that~$\wt(r)\neq 0$ for all~$r\in R$, which we will do from now on. 
Finally, two WTAh are said to be~\emph{equivalent} if they recognize the same tree series.

\begin{example}
Recall the WTAh~$\cA$ and~$\cA'$ from Examples~\ref{ex:first ex}~and~\ref{ex:second ex}
and consider the tree~$k\big(g^2(a),g^3(a)\big)$. The accepting runs~$\varrho$ 
and~$\varrho'$ of~$\cA$ and~$\cA'$, respectively, for it are the following:
	 
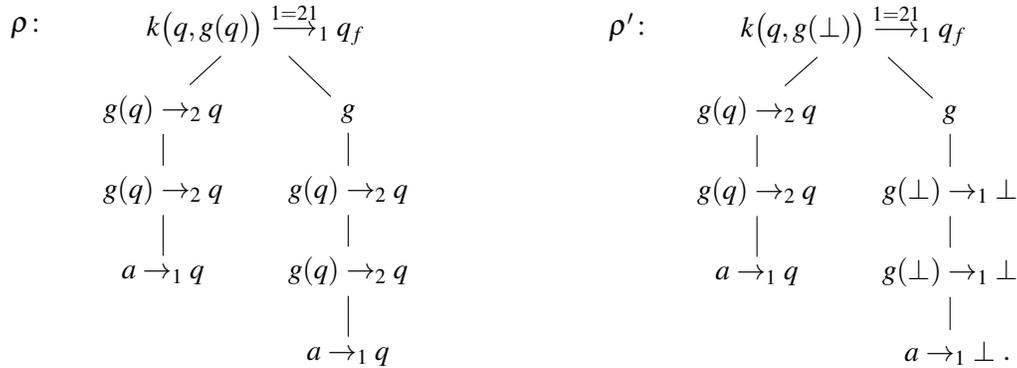
\begin{center}
		\begin{tikzpicture}
			\node at (0,2.085) (a) {$\varrho\colon$};
	\node at (3, 0)  (b) {
				\begin{forest}
					for tree={%
						l sep=0.1cm,
						s sep=0.6cm,
						minimum height=0.000008cm,
						minimum width=0.000015cm,
					}
	[$k\big(q\mathpunct{,} g(q)\big)\stackrel{1=21}\longrightarrow_1 q_f$[$g(q)\to_2 q$
	[$g(q)\to_2 q$[$a\to_1 q$] ] ]
		[$g$[$g(q)\to_2 q$[$g(q)\to_2 q$[$a\to_1 q$] ] ] ] ]
				\end{forest}
			};		
	\node at (8,2.085) (a) {$\varrho'\colon$};
	\node at (11, 0)  (b) {
				\begin{forest}
					for tree={%
						l sep=0.1cm,
						s sep=0.6cm,
						minimum height=0.000008cm,
						minimum width=0.000015cm,
					}
	[$k\big(q\mathpunct{,} g(\bot)\big)\stackrel{1=21}\longrightarrow_1 q_f$[$g(q)\to_2 q$
	[$g(q)\to_2 q$[$a\to_1 q$] ] ]
		[$g$[$g(\bot)\to_1 \bot$[$g(\bot)\to_1 \bot$[$a\to_1 \bot$] ] ] ] ]
				\end{forest}
			};			\node at (13, -2.32) (v) {.};
	\end{tikzpicture} 
\end{center}
It is~$\wt(\varrho)=2^4$ while~$\wt'(\varrho')=2^2$ because in the eq-restricted WTAh~$\cA'$,
every constrained subtree except for
one (pending from position~$1$) is processed exclusively in the state~$\bot$ with weight~$1$.

Both WTAh are unambiguous, so it is impossible for different accepting runs with complementary 
weights to cancel out. Thus for a tree~$t\in T_\Sigma$ it is~$t\in\supp\sem\cA$ iff.\ $\cA$ 
has an accepting run for~$t$, and the same is true for~$\cA'$. In fact, it is \[ \supp\sem\cA \;
\; =\;\;\supp\sem{\cA'}\;\;=\;\;\big\{ k\big(g^n{a},g^{n+1}(a)\big)\mid n\in\N \big\}\,. \]
 
\end{example}

If a tree series is recognized by a WTA, it is called~\emph{regular}, if it is recognized by some 
WTAh, then it is called~\emph{constraint-regular}, and if it is recognized by an eq-restricted 
WTAh, then it is called~\emph{hom-regular}. This choice of name hints at the fact that 
eq-restricted WTAh are tailored to represent homomorphic images of regular tree 
series. For an illustration of this feature, consider the following example. 

\begin{example}\label{ex:hom image}
Let~$\Sigma=\{a^{(0)},g^{(1)},f^{(1)}\}$ and~$A\colon T_\Sigma
\to\N$ defined for every~$s\in T_\Sigma$ by
\begin{align*}
A(s)=\begin{cases} 2^n &\text{if } s= f\big(g^n(a)\big) \\ 0 &\text{else.}
\end{cases}
\end{align*}

A simple WTA recognizing the tree series~$A$ is~$\cA=\big(\{q,q_f\},
\Sigma,\{q_f\},R,\wt\big)$ with the rules and weights~$R=\big\{ a\to_1 q,\; g(q)\to_2 q,\; f(q)
\to_1 q_f \big\}$. 
Consider~$\Delta=\{a^{(0)},g^{(1)},k^{(2)}\}$ and the input-finitary tree homomorphism~$h\colon 
T_\Sigma \to T_\Delta$ 
induced by the mapping~$h(a)=a,\; h(g)=g(x_1)$ and~$h(f)=k\big(x_1,g(x_1)\big)$. The 
homomorphic image~$h_A$ is the tree series given for all~$t\in T_\Delta$ by
\begin{align*}
h_A(t)=\begin{cases} 2^n &\text{if } t= k\big(g^n(a),g^{n+1}(a)\big) \\ 0 &\text{else.}
\end{cases}
\end{align*}
The natural eq-restricted WTAh that recognizes~$h_A$ 
is~$\cA'=\big(\{q,q_f,\bot\},\Delta,\{q_f\},R',\wt'\big)$ from Example~\ref{ex:second ex} 
with
\begin{align*}
R'= \;&\big\{\, a\to_1 q,\quad \;g(q)\to_2 q,\;\;\,\quad k\big(q,g(\bot)\big) \stackrel{1=21}
\longrightarrow_1 q_f\, \big\} \\ 
\cup \;&\{\,a\to_1\bot, \quad g(\bot)\to_1\bot ,
\quad k(\bot,\bot)\to_1\bot\,\}\; .
\end{align*} 
The new rules in~$R'$ are obtained from the rules in~$R$ by applying the tree homomorphism to their 
left-hand sides. The duplicated subtree below~$k$ targets the sink state~$\bot$ 
instead of~$q$ to avoid distorting 
the weight with an additional factor~$2^n$. 
\end{example}

More formally, the following statement was shown in~\cite{WTAc-journal}. We include a condensed 
version of the proof as we will refer to a technical detail below.
\begin{lemma}[cf.~\protect{\cite[Theorem~5]{WTAc-journal}}]\label{lm:WTAh for hom image}
Let~$\SR$ be a commutative semiring, $\cA=\big(Q,\Sigma,F,R,\wt\big)$ a 
WTA over~$\SR$ and~$h\colon T_\Sigma\to T_\Delta$ a 
nondeleting and nonerasing tree homomorphism. There is an eq-restricted WTAh~$\cA'$ that 
recognizes~$h_{\sem\cA}$.
\end{lemma}
\begin{proof}
An eq-restricted WTAh~$\cA'$ for~$h_{\sem\cA}$ is constructed in two stages.  
  
  First, we define~$\cA'' = \bigl(Q \dot{\cup} \{\bot\}, \Delta \cup \Delta
  \times R, F'', R'', \wt'' \bigr)$ such that for every~$r = \sigma(\seq q1k) \to_{\wt(r)} q$ 
  in~$R$ and $h(\sigma) = u =  \delta(\seq u1n)$, we include
  \begin{align*} r'' \;=\; \Bigl( \langle \delta,r\rangle(\seq u1n) \llbracket \seq
    q1k \rrbracket \stackrel{E}\longrightarrow_{\wt''(r'')} q \Bigr) \;\in\;
    R'' \qquad \text{with} \qquad E \;=\; \bigcup_{i \in [k]}
    \pos_{x_i}(u)^2\end{align*}
  where the substitution~$\langle \delta,r\rangle(\seq u1n)\llbracket
  \seq q1k \rrbracket$ replaces for every~$i \in [k]$ only the
  $\leq_{\text{lex}}$-minimal occurrence of~$x_i$ in~$\langle \delta,r\rangle(\seq u1n)$
  by~$q_i$ and all other occurrences by~$\bot$. 
  We set $\wt''(r'')= \wt(r)$.
  Additionally, we let~$r''_\delta = \delta(\bot, \dotsc, \bot) \to \bot \in R''$ with~$\wt''(r''_\delta) 
  = 1$ for every~$k \in \nat$ and $\delta \in \Delta_k$. 
  No other productions are in~$R''$.  Finally, we let~$F''  = F$.
  
  We can now delete the annotation: We use a deterministic relabeling to remove
  the second components of labels of~$\Delta \times R$, adding up the weights of now identical 
  rules. Since hom-regular languages are closed under relabelings \cite[Theorem~4]{WTAc-journal}, we obtain an
  eq-restricted WTAh~$\cA'=\big(Q\dot{\cup}\{\bot\},\Delta,  F',R',\wt'\big)$ recognizing~$h_{\sem\cA}$.
   \end{proof}

The WTAh constructed for the homomorphic image of a WTA preserves the original state
behaviour in its leading copies of duplicated subtrees. Using the notation from the proof of 
Lemma~\ref{lm:WTAh for hom image}, we want to define a mapping that traces the runs of the input 
WTA to its homomorphic image.

\begin{definition}\label{def:hR on runs}
   Let~$\cA,h$ and~$\cA'$ be as in Lemma~\ref{lm:WTAh for hom image}, let~$r=
   \sigma(\seq q1k) \to q\in R$ and~$h(\sigma) =
   \delta(\seq u1n)$. We let~$h^R(r)$ be the rule~$\delta(\seq u1n) \llbracket \seq
    q1k \rrbracket \stackrel{E}\longrightarrow q$ of the WTAh~$\cA'$. 
    
    The assignment~$h^R$ extends naturally to the runs of~$\cA$:
    For a run of the form~$\varrho=r=(\alpha\to q)$ with $\alpha\in\Sigma^0$, we 
    set~$h^R(\varrho)=h^R(r)$. For a run of~$\cA$ of the form~$\varrho=r(\seq {\varrho}1k)$ 
    with~$r=\sigma(\seq q1k)
    \to q$ and~$h(\sigma) = \delta(\seq u1n)$ we set $$h^R(\varrho)=\big(h^R(r)\big)(\seq u1n)
	\llbracket h^R(\varrho_1),\ldots, h^R(\varrho_k)\rrbracket\;;$$
	here, the substitution~$\llbracket h^R(\varrho_1),\ldots, h^R(\varrho_k)\rrbracket$ replaces 
	for every~$i \in [k]$ only the $\leq_{\text{lex}}$-minimal occurrence of~$x_i$ in~$\big(h^R(r)
	\big)(\seq u1n)$ by~$h^R(\varrho_i)$ and all other occurrences by the respective unique run 
	to~$\bot$ for the unique tree that satisfies the constraint~$E$. 
\end{definition}

Using the notation above, the assignment~$h^R\colon R\to R'$ is well-defined, but 
not necessarily injective, and its image is~$h^R(R)=\{r'\in R'\mid r'\text{ targets some }q\neq\bot\}$. 
Let us see how it acts on our running example.

\begin{example}
Recall the WTA~$\cA$ and WTAh~$\cA'$ from Example~\ref{ex:hom image}. The mapping~$h^R$
assigns $$h^R\;\colon \quad f(q)\to q_f \quad\mapsto \quad k\big(q,g(\bot)\big)\stackrel{1=21}
\longrightarrow q_f\; ,$$ and for the unique run of~$\cA$ for the tree~$f\big(g(a)\big)$, it is
	 \begin{center}
		\begin{tikzpicture}
	\node at (-2,0) (label) {$h^R\; \colon$};
	\node at (0, 0)  (a) {
				\begin{forest}
					for tree={%
						l sep=0.1cm,
						s sep=0.6cm,
						minimum height=0.000008cm,
						minimum width=0.000015cm,
					}
	[$f(q)\to q_f$[$g(q)\to q$[$a\to q$] ] ]
				\end{forest}
			};
	\node at (2,0) (label) {$\mapsto$};
	\node at (5.2, 0)  (b) {
				\begin{forest}
					for tree={%
						l sep=0.1cm,
						s sep=0.6cm,
						minimum height=0.000008cm,
						minimum width=0.000015cm,
					}
	[$k\big(q\mathpunct{,} g(\bot)\big)\stackrel{1=21}\longrightarrow q_f$[$g(q)\to q$[$a\to q$] ] 
		[$g$[$g(\bot)\to \bot$[$a\to \bot$] ]] ]
				\end{forest}
			};		
	\node at (7,-1.81) (label) {.};
	\end{tikzpicture} 
\end{center}
\end{example}

When discussing the behaviour of a WTAh~$\cA$, we often argue with the help of 
runs~$\varrho$, so it is a nuisance that we might have~$\wt(\varrho) = 0$.  This 
anomaly can occur even if~$\wt(r) \neq 0$ for all rules~$r$ of~$\cA$ due to the 
presence of zero-divizors, that is, elements~$s, s' \in \SR \setminus\{0\}$ such 
that~$s \cdot s' = 0$. Fortunately, we can avoid this altogether using a
construction of~\cite{kir11}, which is based on~\textsc{Dickson}'s 
Lemma~\cite{dic13}. It was first lifted to tree automata in~\cite{droheu15} and later to WTAh in~\cite{dlt22,WTAc-journal}. Here, we slightly 
adjust the proof of Lemma 3 in~\cite{WTAc-journal} such that it 
preserves the eq-restriction of the input WTAh. 

\begin{lemma}(cf.~\protect{\cite[Lemma~3]{WTAc-journal}})
\label{lm:zero-divizors}
Let~$\SR$ be a commutative semiring. For every eq-restricted WTAh~$\cA$ 
over~$\SR$ there exists an eq-restricted WTAh~$\cA'$ equivalent 
  to~$\cA$ such that~$\wt_{\cA'} (\varrho') \neq 0$ for all runs~$\varrho'$ of~$\cA'$.  
For each~$t\in  \supp\sem\cA$, the accepting (i.e.\ of non-zero weight 
  and targeting a final state) runs of~$\cA$ for~$t$ translate bijectively into the accepting 
  runs~$\varrho'$ of~$\cA'$ for~$t$, and the weights are preserved.
\end{lemma}
\begin{proof} Let $\cA$ be the eq-restricted WTAh~$\big(Q\dot{\cup}\{\bot\}, \Sigma, F, R, 
\wt\big)$.  Obviously, $(\mathbb S, \mathord\cdot, 1, 0)$~is a commutative monoid with zero.  
Let~$(\seq s1n)$ be an enumeration of the finite set~$\wt(R) 
  \setminus \{1\} \subseteq \mathbb S$.  We consider the monoid homomorphism~$h \colon 
  \nat^n \to \mathbb S$, which is given for every~$\seq m1n \in \nat$ by
  \[ h(\seq m1n) = \prod_{i = 1}^n s_i^{m_i}\;. \]
 According to \textsc{Dickson}'s lemma~\cite{dic13}, the
  set~$\min\big( h^{-1}(0)\big)$ is finite, where the partial order is the standard pointwise order 
  on~$\nat^n$.  Hence there is~$u \in\nat$ such that~$\min \big(h^{-1}(0)\big) \subseteq \{0, 
  \dotsc, u\}^n = U$. We define the operation~$\mathord{\oplus} \colon U^2 \to U$ by~$(v
  \oplus v')_i = \min(v_i + v'_i, u)$ for every~$v, v' \in U$ and~$i
  \in [n]$.  Moreover, for every~$i \in [n]$ we let~$1_{s_i} \in U$
  be the vector such that~$(1_{s_i})_i = 1$ and~$(1_{s_i})_a = 0$
  for all~$a \in [n] \setminus \{i\}$.  Let~$V = U \setminus
  h^{-1}(0)$.  We construct the equivalent eq-restricted WTAh~$\cA'=\big(Q'\dot{\cup}\{\bot\},
  \Sigma,F',R',\wt'\big)$ such that $Q' = Q \times V$, $F'= F\times V $, and $R'$~and~$\wt'$ are 
  given as follows.  Consider a rule $r = \ell \stackrel{E}\longrightarrow q \in R $, 
  let~$\pos_Q(\ell)=\{\seq p1k\}$ ordered lexicographically and let~$q_i=\ell(p_i)$ for all~$i\in 
  [k]$. Note that we do not consider the leaves of~$\ell$ that are labeled by~$\bot$. 
  For all choices of~$\seq v1k \in V$ such that the value $v = 1_{\wt(r)} \oplus \bigoplus_{i = 1}^k v_i $
  is again in $V$, the production \[ \ell[\langle q_1, v_1\rangle]_{p_1} \dotsc [\langle q_k, v_k\rangle]_{p_k}
   \stackrel{E}\longrightarrow  \langle q, v \rangle \] belongs to~$R'$ and its weight
  is~$\wt'(p') = \wt(r)$.  No further rules are in~$R'$.  

By annotating the power vectors~$v_i$ to the states~$q\neq \bot$, we suitably (for the purpose of
  zero-divisors) track the weight of runs as~$v$. If attaching another rule adopted from~$R$ to so 
  far valid runs of~$\cA'$ would evaluate the overall weight to zero, then we exclude this rule 
  from~$R'$. Consequently, every run of~$\cA'$ is valid. 
  To preserve the eq-restriction, we only annotate power vectors~$v_i$ to the non-sink states. It is safe 
  to omit~$\bot$ 
  in this construction since~$\bot$ only ever processes the neutral weight~$1$.
\end{proof}	      

From here on, we silently assume that each WTAh avoids zero-divizors.

A main result proved in this article is deciding regularity for unambiguous WTAh over zero-sum free 
commutative semirings. We achieve this by reducing the problem to the unweigted (i.e.\ boolean) case 
solved in~\cite{godoy2013hom}. For this, we must relate our WTAh model to the~\emph{tree automata 
with HOM equality constraints} used by~\cite{godoy2013hom} which differ slightly from our WTAh 
over the boolean semiring. Fortunately, the two are closely related and the translation is rather 
simple: We merely eliminate the sink state and drop the weight assignment.

\begin{lemma}\label{lm:booleanization}
Let~$\SR$ be a commutative semiring and~$\cA=\big(Q\dot\cup 
\{\bot\},\Sigma,F,R,\wt\big)$ an eq-restricted WTAh over~$\SR$. 
If~$\cA$ is unambiguous or~$\SR$ is zero-sum free, then there is a \emph{tree automaton with 
HOM equality constraints (TA$_{hom}$)}~\cite{godoy2013hom}~$\cA^\B$  
that recognizes the tree language~$\supp\sem\cA$. If~$\cA$ is a WTA (i.e.\ without constraints), 
then $\cA^\B$ is also a TA without constraints.
\end{lemma}
\begin{proof}
Let~$q\in Q$ and consider a rule~$\ell\stackrel{E}\longrightarrow_a q$ of~$\cA$. Suppose 
that~$\{p_1^1,
\ldots,p_{n_1}^1\},\; \ldots \; ,\{p_1^m,\ldots, p_{n_m}^m\} $ are the equivalence classes of~$E$, and 
wlog.\ let~$p_1^i$ be the unique representative such that~$\ell(p_1^i)\neq\bot$ for each~$i
\in[m]$. Then we include the unweighted rule $$ \text{\large{$\ell$}}\,[\ell(p_1^1)]_{p^1_2}\cdots 
[\ell(p_1^1)]_{p^1_{n_1}} \;\cdots\; [\ell(p_1^m)]_{p^m_2}\cdots [\ell(p_1^m)]_{p^m_{n_m}}
\stackrel{E}\longrightarrow q $$ in~$R^\B$, that is, we replace every occurrence of~$\bot$ by the 
unique state from~$Q$ that labels a related position. This is necessary because the definition 
of~TA$_{hom}$ requires~$E$-related positions to be labelled with the same state. We 
proceed this way for every rule of~$\cA$, discarding the rules that target~$\bot$, and obtain the 
(unweighted)~TA$_{hom}$~$\cA^\B=\big(Q,\Sigma,F,R^\B\big)$. Since~$\cA$ avoids 
zero-divizors, the conditions in the 
statement are each sufficient to ensure that~$t\in\supp\sem\cA$ iff.\ there exists an run of~$\cA$
for~$t$ to a final state, so~$\cA^\B$ recognizes~$\supp\sem\cA$.
\end{proof}

\begin{example}
Reconsider the WTAh~$\cA'$ from~Example~\ref{ex:hom image}. To obtain the~TA$_{hom}$
$(\cA')^\B$, we remove the sink state~$\bot$, all rules that target~$\bot$ and the weight 
assignment, and replace the rule~$k\big(q,g(\bot)\big)\stackrel{1=21}\longrightarrow_1 q_f$ 
with the unweighted rule~$k\big(q,g(q)\big)\stackrel{1=21}\longrightarrow q_f$.
\end{example}
\section{Deciding Regularity for Unambiguous WTAh}\label{sec:reg decid}

In this section, we prove that regularity is decidable for unambiguous 
eq-restricted WTAh over zero-sum free semirings. To this end, we 
reduce this problem to regularity in the unweighted case, which was 
proved decidable in~\cite{godoy2013hom}.

We begin by defining the~\emph{linearization} of eq-restricted WTAh, which was introduced for the 
boolean case in~\cite{godoy2013hom} and adapted to the weighted model in~\cite{integer-hom}.
The linearization of a WTAh~$\cA$ by the number~$\mathtt{h}$ is a WTG~$\mathtt{lin}(\cA,
\mathtt{h})$ that approximates~$\cA$: It simulates all runs of~$\cA$ which only enforce the 
equality of subtrees of height at most~$\mathtt{h}$.  
This is achieved by instantiating the constrained $Q$-positions of every rule~$\ell \stackrel E
\longrightarrow q$ in~$\cA$ with compatible trees of height at most~$\mathtt{h}$, while the
$Q$-positions of~$\ell$ that are unconstrained by~$E$ remain unchanged. 

Formally, the linearization is defined following~\cite[Definition~7.1]{godoy2013hom}.

\begin{definition}[cf.~\protect{\cite[Definition~12]{integer-hom}}]
 Let~$\SR$ be a commutative semiring. Consider an eq-restricted WTAh~$\cA = (Q \dot\cup \{\bot
 \},  \Sigma, F, R,\wt)$ over~$\SR$, and let~$\mathtt{h}  \in\N$ be a nonnegative integer. 
 The~\emph{linearization  of~$\cA$ by~$\mathtt{h}$} is the WTG $\mathtt{lin}(\cA',\mathtt{h})=
 \big(Q, \Sigma, F,  R_{\mathtt{h}},\wt_{\mathtt{h}}\big)$, 
 where~$R_{\mathtt{h}}$~and~$\wt_{\mathtt{h}}$ are defined  as follows. 
 
 For~$\ell' \in T_\Sigma(Q\dot\cup\{\bot\}) $ and~$q \in Q$, we include the rule~$(\ell' \to q) $ in $R_{\mathtt{h}}$ iff.\ there exist a rule~$(\ell 
 \stackrel E\longrightarrow q) \in R$, positions~$\seq p1k \in \pos_{Q \dot\cup \{\bot\}}(\ell)$, 
 and trees~$\seq t1k \in T_\Sigma$ such that
  \begin{itemize}
  \item $\{\seq p1k\} = \bigcup_{p \in \pos_\bot(\ell)} [p]_E$, that is, $E$~constrains exactly the 
  positions~$\seq p1k$,
  \item $(p_i, p_j) \in E$ implies $t_i = t_j$ for all~$i,j \in[k]$,
  \item $\ell' = \ell[t_1]_{p_1} \dotsm [t_k]_{p_k}$, and
  \item $\wt^{\ell(p_i)}(t_i) \neq 0$ and~$\he(t_i) \leq\mathtt{h}$ for all~$i \in [k]$.
  \end{itemize}
  For every such production~$\ell' \to q$ we set~$\wt_{\mathtt{h}}(\ell' \to q)$ 
  as the sum over all weights
  \[ \wt(\ell \stackrel E\longrightarrow q) \cdot \prod_{i \in
      [k]} \wt^{\ell({p_i})}(t_i) \] for all~$(\ell \stackrel  E\longrightarrow q) \in R$, 
      $\;\seq p1k \in \pos_{Q \dot\cup\{\bot\}}(\ell)$ and~$\seq t1k \in T_\Sigma$ as above.
\end{definition}

Note that the linearization is a WTG without constraints, so it recognizes a regular tree series. 
Let us apply this construction to our running example.

\begin{example}
We recall the WTAh~$\cA'$ from Example~\ref{ex:hom image} and set~$\mathtt{h}=2$. 
The linearization of~$\cA'$ by~$2$ instantiates every constrained 
position by compatible trees of maximal height~$2$, keeping track of the weights, 
and removes~$\bot$ and the rules that target it. More precisely, $\mathtt{lin}(\cA',2)=\big(\{q,q_f\},
\Delta,\{q_f\},R_2,\wt_2\big)$ with the set of rules and weights 
\begin{align*} 
R_2\; =\; \big\{\, &a\to_1 q,\qquad g(q)\to_2 q, \qquad k\big(a,g(a)\big) \to_1 q_f,\\
	&k\big(g(a),g(g(a))\big) \to_2 q_f, \qquad k\big(g(g(a)),g(g(g(a)))\big) \to_4 q_f\;\big\}.
\end{align*}
\end{example}

This example illustrates that the larger we choose~$\mathtt{h}$, the better~$\mathtt{lin}(\cA',
\mathtt{h})$ approximates~$\sem{\cA'}$. In this particular case however, there will always be a 
tree~$t$ such that~$\sem{\cA'}(t)\neq\sem{\mathtt{lin}(\cA',\mathtt{h})}(t)$, say, the tree~$k\big(
g^{\mathtt{h}+1}(a),g^{\mathtt{h}+2}(a)\big)$. For eq-restricted WTAh~$\cA$ over~$\B$ 
or~$\N$ it is known~\cite{godoy2013hom,WTAc-journal} that~$\sem\cA$ is regular iff.\ 
$\sem{\mathtt{lin}(\cA,\mathtt{h})}=\sem\cA$ for a certain parameter~$\mathtt{h}$. For other
semirings, a customized investigation is necessary, but unambiguous WTAh allow us to decide 
regularity by applying the boolean case directly. To this end, the following lemma is fundamental.

\begin{lemma}\label{lm:lin is unambiguous}
Let~$\SR$ be a commutative semiring, $\cA$ an eq-restricted WTAh 
over~$\SR$ and~$h\in\N$. For each $t\in\supp\sem\cA$, there are 
at most as many accepting runs of~$\mathtt{lin}(\cA,\mathtt{h})$ for~$t$ 
as there are accepting 
runs of~$\cA$ for~$t$. In particular, if~$\cA$ is unambiguous, then 
so is its linearization, and
for every~$t\in\supp\sem\cA$
it is either~$\sem{\mathtt{lin}(\cA,\mathtt{h})}(t)=\sem\cA(t)$, or there are no accepting runs of~$
\mathtt{lin}(\cA,\mathtt{h})$ for~$t$.
\end{lemma}
\begin{proof}
The linearization~$\mathtt{lin}(\cA,\mathtt{h})$ is defined in such a way that it simulates every 
run~$\varrho$ of~$\cA$ with the following property: Say~$\varrho$ processes~$t\in T_\Sigma$, 
then for every rule~$\ell\stackrel{E}
\longrightarrow q$ used in~$\varrho$ at position~$p$ (that is, at the root of~$t|_p$), 
and for every position~$\bar{p}$ constrained by~$E$, it is~$\he(t|_{p\bar{p}})\leq \mathtt{h}$.
Different runs of~$\cA$ might be merged into the same run of~$\mathtt{lin}(\cA,\mathtt{h})$, but 
for a particular run of~$\cA$ it is uniquely determined which run of~$\mathtt{lin}(\cA,\mathtt{h})$ 
will incorporate it.
\end{proof}

We need yet another technical ingredient for the reduction to the unweighted 
case: to interchange the linearization of a WTAh and 
its projection onto the boolean TA$_{hom}$.
The linearization for TA$_{hom}$ was defined in~\cite[Definition~7]{godoy2013hom} and indeed, the 
following holds.
\begin{lemma}\label{lm:linB-Blin}
Consider an unambiguous, eq-restricted WTAh~$\cA$ over a commutative semiring. Let~$\cA^\B$ the 
TA$_{hom}$ for~$\supp\sem\cA$ 
defined in Lemma~\ref{lm:booleanization} and~\emph{linearize}$(\cA^\B,\mathtt{h})$ in turn the 
linearization of~$\cA^\B$ by~$\mathtt{h}$ as introduced in~\cite[Definition~7.1]{godoy2013hom}. 
Then it is~$\mathtt{lin}(\cA,\mathtt{h})^\B =\text{\emph{linearize}}(\cA^\B,\mathtt{h})$.
\end{lemma}
We are now ready for the main result of this section: the reduction of regularity for eq-restricted WTAh 
over zero-sum free semirings to the unweighted case.

\begin{theorem}\label{thm:regularity}
Let~$\SR$ be a zero-sum free commutative semiring and~$\cA$ an unambiguous 
eq-restricted WTAh over~$\SR$. The tree series~$\sem\cA$ is regular iff.\
$\supp\sem\cA$ is a regular tree language.
\end{theorem}
\begin{proof}
Suppose first that~$\sem\cA$ is regular, thus there is a WTA~$\cB$ equivalent to~$\cA$. 
Since~$\SR$ is zero-sum free, we can apply Lemma~\ref{lm:booleanization} to~$\cB$ and obtain 
that~$\supp\sem\cB=\supp\sem\cA$ is regular.

Next, suppose that~$\sem\cA$ is not regular. In particular, the regular WTG $\mathtt{lin}(\cA,
\mathtt{h})$ is not equivalent to~$\cA$ for any~$h\in \N$. Thus by 
Lemma~\ref{lm:lin is unambiguous}, it is~$\supp\sem\cA\neq \supp\sem{\mathtt{lin}(\cA,
\mathtt{h})}$. By Lemma~\ref{lm:booleanization}, $\mathtt{lin}(\cA,\mathtt{h})^\B$ recognizes the 
regular language~$\supp{\sem{\mathtt{lin}(\cA,\mathtt{h})}}$, and together with 
Lemma~\ref{lm:linB-Blin}, it is $$\sem{\cA^\B}=\supp\sem\cA\neq\sem{\mathtt{lin}(\cA,
\mathtt{h})^\B}=\sem{\text{\emph{linearize}}(\cA^\B,\mathtt{h})},$$ 
that is, the boolean linearization of the TA$_{hom}$~$\cA^\B$ is not equivalent 
to it for any~$h\in\N$. This, however, implies that~$\sem{\cA^\B}=\supp\sem\cA$ is not regular 
as proved in~\cite[Lemma~7.3]{WTAc-journal}.
\end{proof}
Note that we only used zero-sum freeness of the semiring for the first part of the statement, as 
Lemma~\ref{lm:booleanization} holds for unambiguous WTAh over arbitrary commutative semirings.
With this result, regularity of eq-restricted WTAh is decidable.
\begin{corollary}\label{cor:reg decid}
Let~$\SR$ be a zero-sum free commutative semiring. Given an unambiguous eq-restricted 
WTAh~$\cA$ over~$\SR$ as input, it is decidable whether~$\sem\cA$ is regular.
\end{corollary}
\begin{proof}
By Theorem~\ref{thm:regularity}, $\sem\cA$ is regular iff.\ $\supp\sem\cA$ is regular. A 
TA$_{hom}$ recognizing the latter can be constructed with Lemma~\ref{lm:booleanization}, for 
which, in turn, regularity is decidable~\cite[Section~7]{godoy2013hom}.
\end{proof}

\section{A Sufficient Condition and the HOM-Problem}\label{sec:sufficient}
So far, the assumption we make for deciding regularity is imposed on the WTAh. Meanwhile the 
HOM-problem has a WTA~$\cA$ and a tree homomorphism~$h$ as input. In this 
section, we propose conditions on~$\cA$ and~$h$ which ensure that the strategy of the previous 
section is applicable to the corresponding instance of the HOM-problem. We begin with a condition 
for~$h$ which generalizes injectivity.

\begin{definition}
Let~$\Sigma$ and~$\Delta$ be ranked alphabets and~$h\colon T_\Sigma\to T_\Delta$ a 
nondeleting and nonerasing tree homomorphism. We call~$h$~\emph{tetris-free} if for all~$s,s'
\in T_\Sigma$ with~$h(s)=h(s')$, it is~$\pos(s)=\pos(s')$ and for all~$p\in\pos(s)$, we 
have~$h\big(s(p)\big)=h\big(s'(p)\big)$.
\end{definition}
\setcounter{figure}{0}    

In other words, $h\colon T_\Sigma\to T_\Delta$ is tetris-free if we cannot combine the building 
blocks~$h(\sigma),\; \sigma\in\Sigma$ in different ways to build the same tree. In contrast, 
Figure~1 below shows the well-known  \emph{Tetriminos\textsuperscript{\textregistered}}\textsuperscript{\textcopyright}~\cite{tetris} violating (and thus naming) the tetris-free condition.

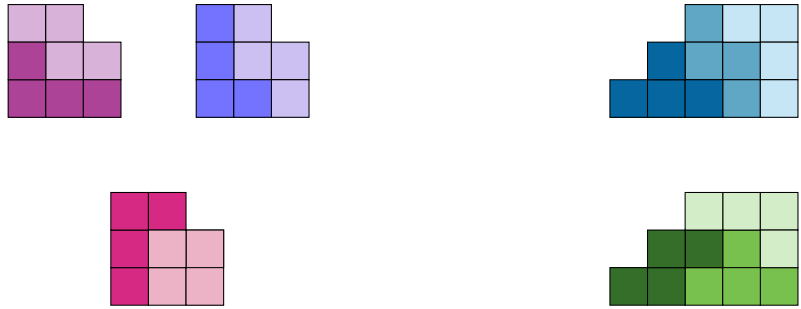
\begin{figure}\label{fig:tetris}
\begin{center}
\begin{tikzpicture}[scale=0.5]
  \begin{scope}[shift={(2.73, -8)}] 
  \begin{scope}[rotate=270]
  \draw[fill=white!20!purple!80!magenta] (-3, 0) rectangle (-4, 1);
  \draw[fill=white!20!purple!80!magenta] (-4, 0) rectangle (-5, 1);
  \draw[fill=white!20!purple!80!magenta] (-5, 0) rectangle (-6, 1);
  \draw[fill=white!20!purple!80!magenta] (-5, 1) rectangle (-6, 2);

  \draw[fill=purple!30!white] (-3, 1) rectangle (-4, 2);
  \draw[fill=purple!30!white] (-3, 2) rectangle (-5, 3);
  \draw[fill=purple!30!white] (-4, 2) rectangle (-5, 3);
  \draw[fill=purple!30!white] (-4, 1) rectangle (-5, 2);
  \end{scope} \end{scope}

  \begin{scope}[shift={(0, 3)}]
  \begin{scope}[rotate=270]
  \draw[fill=violet!30!white] (0, 0) rectangle (1, 1);
  \draw[fill=violet!30!white] (0, 1) rectangle (1, 2);
  \draw[fill=violet!30!white] (1, 1) rectangle (2, 2);
  \draw[fill=violet!30!white] (1, 2) rectangle (2, 3);

  \draw[fill=violet!65!pink] (2, 0) rectangle (3, 1);
  \draw[fill=violet!65!pink] (2, 1) rectangle (3, 2);
  \draw[fill=violet!65!pink] (2, 2) rectangle (3, 3);
  \draw[fill=violet!65!pink] (1, 0) rectangle (2, 1);
  \end{scope} \end{scope}

  \begin{scope}[shift={(5, 9)}] 
  \begin{scope}[rotate=270]
  \draw[fill=white!75!pink!80!blue] (6, 1) rectangle (7, 2);
  \draw[fill=white!75!pink!80!blue] (7, 1) rectangle (8, 2);
  \draw[fill=white!75!pink!80!blue] (7, 2) rectangle (8, 3);
  \draw[fill=white!75!pink!80!blue] (8, 2) rectangle (9, 3);

  \draw[fill=white!45!blue] (6, 0) rectangle (7, 1);
  \draw[fill=white!45!blue] (7, 0) rectangle (8, 1);
  \draw[fill=white!45!blue] (8, 0) rectangle (9, 1);
  \draw[fill=white!45!blue] (8, 1) rectangle (9, 2);
  \end{scope} \end{scope}
  
  \begin{scope}[shift={(10, 0)}] 
  \draw[fill=white!50!gray!55!cyan!70!teal] (8, 2) rectangle (9, 3);
  \draw[fill=white!50!gray!55!cyan!70!teal] (8, 1) rectangle (9, 2);
  \draw[fill=white!50!gray!55!cyan!70!teal] (9, 1) rectangle (10, 2);
  \draw[fill=white!50!gray!55!cyan!70!teal] (9, 0) rectangle (10, 1);

  \draw[fill=white!10!blue!25!teal] (6, 0) rectangle (7, 1);
  \draw[fill=white!10!blue!25!teal] (7, 0) rectangle (8, 1);
  \draw[fill=white!10!blue!25!teal] (8, 0) rectangle (9, 1);
  \draw[fill=white!10!blue!25!teal] (7, 1) rectangle (8, 2);
  
  \draw[fill=white!83!cyan!95!teal] (9, 2) rectangle (10, 3);
  \draw[fill=white!83!cyan!95!teal] (10, 2) rectangle (11, 3);
  \draw[fill=white!83!cyan!95!teal] (10, 1) rectangle (11, 2);
  \draw[fill=white!83!cyan!95!teal] (10, 0) rectangle (11, 1);
  \end{scope}
  
  \begin{scope}[shift={(10, -5)}] 
  \draw[fill=white!30!gray!38!teal!55!lime] (8, 0) rectangle (9, 1);
  \draw[fill=white!30!gray!38!teal!55!lime] (10, 0) rectangle (11, 1);
  \draw[fill=white!30!gray!38!teal!55!lime] (9, 1) rectangle (10, 2);
  \draw[fill=white!30!gray!38!teal!55!lime] (9, 0) rectangle (10, 1);

  \draw[fill=black!60!lime!68!teal] (6, 0) rectangle (7, 1);
  \draw[fill=black!60!lime!68!teal] (7, 0) rectangle (8, 1);
  \draw[fill=black!60!lime!68!teal] (8, 1) rectangle (9, 2);
  \draw[fill=black!60!lime!68!teal] (7, 1) rectangle (8, 2);
  
  \draw[fill=white!83!lime!86!teal] (9, 2) rectangle (10, 3);
  \draw[fill=white!83!lime!86!teal] (10, 2) rectangle (11, 3);
  \draw[fill=white!83!lime!86!teal] (10, 1) rectangle (11, 2);
  \draw[fill=white!83!lime!86!teal] (8, 2) rectangle (9, 3);
  \end{scope}

\end{tikzpicture}
\end{center}
\caption{The game of \emph{Tetris}\textsuperscript{\textregistered}\textsuperscript{\textcopyright}~\cite{tetris} being non-tetris-free by nature.}
\end{figure}

Let us discuss a short example and counter-example.

\begin{example}\label{ex:tetris}
Let~$\Sigma=\{a^{(0)},b^{(0)},g^{(1)}\}$ and~$\Delta=\{c^{(0)},k^{(2)}\}$. Consider the tree 
homomorphism~$h\colon T_\Sigma\to T_\Delta$ induced by~$h(a)=h(b)=c$ 
and~$h(g)=k(x_1,x_1)$. While~$h$ is not injective, it is tetris-free.
However, the tree homomorphism~$\hat{h}\colon T_\Sigma\to T_\Delta$ induced by~$\hat{h}(a)
=c$, $\hat{h}(b)=k(c,c)$ and~$\hat{h}(g)=k(x_1,c)$ is not: The trees~$g(a)$ and~$b$ violate the 
tetris-free condition.
\end{example}

Intuitively, if a tree homomorphism~$h$ is tetris-free, then any non-injective behaviour of~$h$ is 
located entirely 
at the symbol level. This allows the construction of the WTAh to cancel the non-injectivity of~$h$.
For this, however, we also need to make an assumption on the input WTA~$\cA$, which leads us 
to this augmented version of unambiguity for~$\cA$.

\begin{definition}
Let~$\cA$ be a WTA over a commutative semiring~$\SR$ and~$\Sigma$, and~$h\colon T_\Sigma
\to T_\Delta$ a nondeleting and nonerasing, tetris-free tree homomorphism. We say that~$\cA$ 
is~\emph{$h$-unambiguous} if for all trees~$s,s'\in T_\Sigma$ such that~$h(s)=h(s')$, all 
accepting runs~$\varrho,\varrho'$ of~$\cA$ for~$s$ and~$s'$, respectively, and all~$p
\in\pos(s)$, the target states of the rules applied in~$\varrho$ and~$\varrho'$ at~$p$, 
respectively, coincide.
\end{definition}

\begin{remark}\label{rm:tetris-free and h-unamb}
The condition of $h$-unambiguity is stronger than unambiguity: For~$s=s'\in\supp\sem\cA$ we 
obtain that~$\cA$ has at most one accepting run for~$s$ (since runs of WTA are uniquely 
determined by the processed symbol and the target state at every position). A similar reasoning applies 
if we choose~$s\neq s'$ with~$h(s)=h(s')$: While the unique runs of~$\cA$ for~$s$ and~$s'$ may
read different symbols, the states they pass through coincide at every position.
\end{remark}

Imposing these conditions on the input of the HOM-problem allows us to 
build on it with the arguments from the previous section. 

\begin{proposition}\label{prop:tetris and unamb}
Let~$\cA$ be a WTA over a commutative semiring~$\SR$ and~$\Sigma$, and~$h\colon T_\Sigma
\to T_\Delta$ a nondeleting and nonerasing tree homomorphism. If~$h$ is tetris-free and~$\cA$
is~$h$-unambiguous, then the eq-restricted WTAh~$\cA'$ for~$h_{\sem\cA}$ constructed in 
Lemma~\ref{lm:WTAh for hom image} is unambiguous.
\end{proposition}
\begin{proof}
Let~$\vartheta$ and~$\vartheta'$ be accepting runs of~$\cA'$ for the same~$t\in T_\Delta$.
We prove the statement by contradiction, so assume that~$\vartheta\neq\vartheta'$. Then there 
are two distinct runs~$\varrho$ and~$\varrho'$ of~$\cA$ such that~$\vartheta=h^R(\varrho)$ 
and~$\vartheta'=h^R(\varrho')$ as introduced in Definition~\ref{def:hR on runs}. The 
mapping~$h^R$ does not modify the target states of runs, so both~$\varrho$ and~$\varrho'$ are 
accepting as well, and since~$\cA$ is unambiguous, they must process distinct trees~$s$ and~$s'$, 
respectively, with~$h(s)=h(s')$. By the premises of the statement, at every~$p\in \pos(s)=\pos(s')$ it 
is~$h\big(s(p)\big)=h\big(s'(p)\big)$, and the target states of~$\varrho$ and~$\varrho'$ at~$p$ 
coincide, so although~$\varrho\neq\varrho'$, after applying~$h$ it 
is~$\vartheta=h^R(\varrho)=h^R(\varrho')=\vartheta'$, which contradicts our assumption.
\end{proof}

We want to illustrate the role played by our two conditions, the $h$-unambiguity and the tetris-freeness. 
Let us discuss this with the help of two counter-examples.

\begin{example}
Consider the ranked alphabets~$\Sigma=\{a^{(0)},b^{(0)},g^{(1)}\}$ and~$\Delta=\{c^{(0)},k^{(2)}\}$.
Let~$h\colon T_\Sigma\to T_\Delta$ be the tetris-free tree homomorphism from 
Example~\ref{ex:tetris} induced 
by~$h(a)=h(b)=c$ and~$h(g)=k(x_1,x_1)$. Moreover, let~$\cA=\big(Q,\Sigma,
Q,R,\wt\big)$ be the WTA over the arctic~semiring~$\mathbb{A}$ with~$Q=\{q_a, q_b\}$ 
and the following rules and weights:
\begin{align*}
R\;=\;\big\{\; a\to_0 q_a\,,\quad b\to_0 q_b\,,\quad g(q_a)\to_1 q_a\,,\quad g(q_b)\to_2 q_b
\;\big\}\; .
\end{align*}
The WTA~$\cA$ is unambiguous, but not~$h$-unambiguous, since the runs for~$a$ and~$b$
target different states despite~$h(a)=h(b)$. Evaluating the weights in~$\mathbb{A}$, we 
obtain the tree series~$\sem\cA$ defined by 
\begin{align*}
\sem\cA\;\colon \;s\;\mapsto \;\begin{cases} n &\text{ if } s=g^n(a) \\ 2n &\text{ if } 
s=g^n(b) \end{cases}
\end{align*}

The WTAh~$\cA'=\big(Q\dot{\cup}\{\bot\},\Delta,Q,R',\wt'\big)$ recognizing~$h_{\sem\cA}$ which is 
obtained from Lemma~\ref{lm:WTAh for hom image} has the following rules and weights:
\begin{align*}
R\;&=\;\big\{\; c\to_0 q_a\,,\quad c\to_0 q_b\,,\quad k(q_a,\bot)\stackrel{1=2}\longrightarrow_1 
q_a\,,\quad k(q_b,\bot)\stackrel{1=2}\longrightarrow_2 q_b\;\big\} \\
&\,\cup\;\:\big\{\;c\to_0 \bot\,,\:\quad k(\bot,\bot)\to_0 \bot
\;\big\}\; .
\end{align*}
Because of the different target states, the rules~$c\to q_a$ and~$c\to q_b$ are not merged 
in~$\cA'$, therefore~$\cA'$ is not unambiguous.

On the other hand, let~$\hat{h}$ be the homomorphism from Example~\ref{ex:tetris} induced 
by~$\hat{h}(a)=c$, $\hat{h}(b)=k(c,c)$ and~$\hat{h}(g)=k(x_1,c)$. Recall that~$h$ is not tetris-free 
because~$h\big(g(a)\big)=h(b)$ although~$\pos\big(g(a)\big)\neq \pos(b)$. Moreover, 
consider the WTA~$\hat\cA=\big(\{q\},\Sigma,\{q\},\hat R,\hat\wt\big)$  over~$\N$ with the 
following rules and weights:
\begin{align*}
\hat R\;=\;\big\{\; a\to_2 q\,,\quad b\to_3 q\,,\quad g(q)\to_1 q \;\big\}\; .
\end{align*}

The WTA~$\hat{\cA}$ only has one state, so it is deterministic and thus unambiguous. It 
recognizes the tree series~$\sem{\hat\cA}$ defined by
\begin{align*}
\sem{\hat\cA}\;\colon \;s\;\mapsto \;2\abs{\pos_a(s)} + 3\abs{\pos_b(s)}\; .
\end{align*}
However, the WTAh~$\hat\cA\,'=\big(\{q,\bot\},\Delta,\{q\},\hat R',\hat\wt'\big)$ 
for~$h_{\sem\cA}$ has the following rules and weights:
\begin{align*}
\hat R'\;&=\;\big\{\; c\to_2 q\,,\quad \:\,k(c,c)\to_3 q\,,\quad k(q,c)\to_1 q\;\big\} \\
&\:\cup\;\,\big\{\;c\to_1 \bot\,,\quad k(\bot,\bot)\to_1 \bot
\;\big\}\,.
\end{align*}
Since~$\hat h$ performs no duplications, the rules targeting~$\bot$ are not used in any 
accepting run, so we can safely ignore them. 
Although this time, no two rules of~$\hat\cA\,'$ (that are used in an accepting run) share a 
left-hand side, the tree $k(c,c)= \hat h \big(g(a)\big)=\hat h(b)$ still has two different runs, which 
stem directly from the non-tetris-freeness of~$\hat h$.
\end{example}
As a concequence of Proposition~\ref{prop:tetris and unamb}, our restricted version 
of the HOM-problem is decidable.

\begin{corollary}\label{cor:hom decid}
Let~$\SR$ be a zero-sum free, commutative semiring. For a nondeleting and nonerasing, 
tetris-free tree homomorphism~$h$ and an~$h$-unambiguous WTA~$\cA$ over~$\SR$ as input, it 
is decidable whether the tree series~$h_{\sem\cA}$ is regular.
\end{corollary} 

\section{Conclusion and Future Work}\label{sec:conclusion}
 
Homomorphic images of regular tree series can be represented using an extension of WTA, the 
so-called \emph{eq-restricted WTAh}~\cite{WTAc-journal}. In this paper, we have shown 
that regularity is decidable for unambiguous devices of this type over zero-sum free commutative 
semirings. For this, we reduced this question to the unweighted setting, where regularity is 
known to be decidable~\cite{godoy2013hom}. Moreover, we have phrased a condition on the input 
WTA~$\cA$ and tree homomorphism~$h$ that ensures unambiguity of the WTAh representing the 
image~$h_{\sem\cA}$. Thus the~\emph{HOM-problem} over zero-sum free semirings which, 
given~$\cA$ and~$h$ as input, asks whether~$h_{\sem\cA}$ is regular, is decidable if the input 
satisfies our condition.

Notably, the zero-sum freeness of the semiring is only used in Theorem~\ref{thm:regularity} to 
show that if the tree series recognized by an unambiguous eq-restricted WTAh~$\cA$ is regular, 
then its support is also regular. It is plausible that the zero-sum freeness is not needed: Its 
purpose is to ensure that different accepting runs of~$\cA$ for the same tree~$t$ cannot cancel out, 
leaving~$t\notin\supp{\sem\cA}$ despite~$\cA$ having accepting runs for~$t$. This, however, 
should not be a concern if~$\cA$ is unambiguous. To discard the zero-sum freeness assumption, it 
suffices to prove this simple statement: \emph{If~$\cA$ is an unambiguous eq-restricted WTAh 
and~$\sem\cA$ is regular, then there is an unambiguous WTA equivalent to~$\cA$.} In fact, the 
linearization 
of~$\cA$, which is unambiguous by Lemma~\ref{lm:lin is unambiguous}, is a promising candidate. 
Thus we conjecture that Theorem~\ref{thm:regularity} holds for arbitrary commutative semirings, as 
do then Corollaries~\ref{cor:reg decid} -- stating that regularity is decidable for unambiguous 
eq-restricted WTAh -- and~\ref{cor:hom decid} -- stating that the HOM-problem is decidable under 
our assumptions on~$\cA$ and~$h$.

Recently, the disambiguation of weighted (string) automata from~\cite{mohri2017} was lifted to 
trees~\cite{stierulbricht}. Here, the authors assume a variation of the~\emph{twins property} 
which restricts the behaviour of related states of a WTA. This allows them to construct an 
equivalent unambiguous WTA. A natural question is whether this proof can be adjusted to provide 
even an $h$-unambiguous WTA, say, by refining the twins property with respect to~$h$. That 
way, we could lift our result to a larger class of input WTA.
\bibliographystyle{eptcs}
\bibliography{references.bib}
\end{document}